\documentclass[journal,10pt,draftclsnofoot,onecolumn]{IEEEtran}

\usepackage[utf8]{inputenc}
\usepackage{amsmath}
\usepackage{amsthm}
\usepackage{graphicx}
\usepackage[nice]{nicefrac}

\usepackage{caption}
\usepackage{subcaption}

\usepackage{xcolor}

\usepackage{tikz}
\usetikzlibrary{arrows}
\usetikzlibrary{shapes}
\usetikzlibrary{backgrounds}

\usepackage{siunitx}
\usepackage{booktabs}

\usepackage{pgfplots}
\usepgfplotslibrary{colormaps}

\newtheorem{lemma}{Lemma} 
\newtheorem{theorem}{Theorem}

\newtheorem{proposition}[lemma]{Proposition} 
\theoremstyle{definition}

\theoremstyle{remark}
\newtheorem{example}{Example}

\newcommand{\Acal}{\mathcal{A}}
\newcommand{\Scal}{\mathcal{S}}

\newcommand{\ot}{\leftarrow}

\author{
	\IEEEauthorblockN{Johannes Rauh\IEEEauthorrefmark{1}, Pradeep Kr. Banerjee\IEEEauthorrefmark{1}, Eckehard Olbrich\IEEEauthorrefmark{1}, J{\"u}rgen Jost\IEEEauthorrefmark{1},\\ Nils Bertschinger\IEEEauthorrefmark{2}, and David Wolpert\IEEEauthorrefmark{3}\IEEEauthorrefmark{4}}\\
	\IEEEauthorblockA{\IEEEauthorrefmark{1}Max Planck Institute for Mathematics in the Sciences, Leipzig, Germany
		\\\{jrauh,pradeep,olbrich,jjost\}@mis.mpg.de}\\
	\IEEEauthorblockA{\IEEEauthorrefmark{2}Frankfurt Institute for Advanced Studies, Frankfurt, Germany
		\\bertschinger@fias.uni-frankfurt.de}\\
	\IEEEauthorblockA{\IEEEauthorrefmark{3}Santa Fe Institute, Santa Fe, NM, USA\\ 
	                  \IEEEauthorrefmark{4}MIT, Cambridge, MA, USA
		              \\dhw@santafe.edu}
}

\title{Coarse-graining and the Blackwell Order}

\begin{document}
\maketitle

% Abstract (Do not use inserted blank lines, i.e. \\) 
\begin{abstract}
%\abstract{
Suppose we have a pair of information channels,~$\kappa_{1},\kappa_{2}$, with a common input. The \emph{Blackwell order} is a partial order over channels that compares~$\kappa_{1}$ and~$\kappa_{2}$ by the maximal expected utility an agent can obtain when decisions are based on the channel outputs.  Equivalently, $\kappa_{1}$ is said to be Blackwell-inferior to~$\kappa_{2}$ if and only if~$\kappa_{1}$ can be constructed by \emph{garbling} the output of~$\kappa_{2}$. A related partial order stipulates that~$\kappa_{2}$ is \emph{more capable} than~$\kappa_{1}$ if the mutual information between the input and output is larger for~$\kappa_{2}$ than for~$\kappa_{1}$ for any distribution over inputs. A Blackwell-inferior channel is necessarily less capable. However, examples are known where~$\kappa_{1}$ is less capable than~$\kappa_{2}$ but not Blackwell-inferior. We show that this may even happen when~$\kappa_{1}$ is constructed by coarse-graining the inputs of~$\kappa_{2}$. Such a coarse-graining is a special kind of ``pre-garbling'' of the channel inputs. This example directly establishes that the expected value of the shared utility function for the coarse-grained channel is larger than it is for the non-coarse-grained channel. This contradicts the intuition that coarse-graining can only destroy information and lead to inferior channels. We also discuss our results in the context of information decompositions. 
%}
\end{abstract}
% Keywords

% The fields PACS, MSC, and JEL may be left empty or commented out if not applicable
%\PACS{J0101}
%\MSC{62C05; %  	General considerations
%  62B15; %  	Theory of statistical experiments
%  94A15 } %  	Information theory, general
%\JEL{}

\section{Introduction}
\label{sec:intro}

Suppose we are given the choice of two channels that both provide information about the same random variable, and that we want to make a decision based on the channel outputs. Suppose that our utility function depends on the joint value of the input to the channel and our resultant decision based on the channel outputs. Suppose as well that we know the precise conditional distributions defining the channels, and the distribution over channel inputs.
Which channel should we choose? The answer to this question depends on the choice of our utility function as well as on the details of the channels and the input distribution. So for example, without specifying how we will use the channels, \emph{in general} we cannot just compare their information capacities to choose between them.

Nonetheless, for certain pairs of channels we can make our choice, even \emph{without} knowing the utility functions or the distribution over inputs. Let us represent the two channels by two (column) stochastic matrices $\kappa_{1}$ and  $\kappa_{2}$, respectively. 
Then if there exists another stochastic matrix~$\lambda$ such that~$\kappa_{1} = \lambda\cdot\kappa_{2}$, there is never any reason to strictly prefer~$\kappa_1$;  for if we choose~$\kappa_2$, we can always make our decision by chaining the output of~$\kappa_{2}$ through the channel~$\lambda$ and then using the same decision
function we would have used had we chosen~$\kappa_1$. 
This simple argument shows that whatever the three stochastic matrices are and whatever the decision rule we would use if we chose channel~$\kappa_1$, we can always get the same expected utility % (as we would get by choosing~$\kappa_1$ and using the associated decision rule)
by instead choosing channel~$\kappa_2$ with an appropriate decision rule.
In this kind of situation, where~$\kappa_{1} = \lambda\cdot\kappa_{2}$, we say that~$\kappa_{1}$ is a \emph{garbling} (or \emph{degradation}) of~$\kappa_{2}$. 
It is much more difficult to prove that the converse also holds true:

\begin{theorem}[Blackwell's theorem~\cite{Blackwell1953}]
  \label{thm:Blackwell}
  Let~$\kappa_{1},\kappa_{2}$ be two stochastic matrices representing two channels with the same input alphabet.  Then the following two conditions are equivalent:
  \begin{enumerate}
  \item When the agent chooses $\kappa_{2}$ (and uses the decision rule 
  that is optimal for~$\kappa_{2}$), her expected utility is always at least as big as the expected utility when she chooses~$\kappa_{1}$ (and uses the optimal decision rule for~$\kappa_{1}$), independent of the utility function and the distribution of the input~$S$.
  \item $\kappa_{1}$ is a garbling of~$\kappa_{2}$.
  \end{enumerate}
\end{theorem}
Blackwell formulated his result in terms of a statistical decision maker who reacts to the outcome of a
\emph{statistical experiment}.  We prefer to speak of a decision problem instead of a statistical experiment.
See~\cite{Torgersen91:Comparison_of_statistical_experiments,LeCam96:Comparison_of_experiments} for an overview.

Blackwell's theorem motivates looking at the following partial order over channels $\kappa_{1},\kappa_{2}$ with a common input alphabet:
\begin{equation*}
  \kappa_{1}\preceq \kappa_{2} %\\
  \quad:\Longleftrightarrow\;
  \begin{cases}
    \text{one of the two statements}\\\text{in Blackwell's theorem holds true.}
  \end{cases}
\end{equation*}
We call this partial order the \emph{Blackwell order} (this partial order is called \emph{degradation order} by other authors \cite{Bergmans1973,CohenKempermanZbaganu98:Comparison_of_Stochastic_Matrices}). If $\kappa_{1}\preceq \kappa_{2}$, then $\kappa_{1}$ is said to be Blackwell-inferior to~$\kappa_{2}$.
Strictly speaking, the Blackwell order is only a preorder, since there are channels $\kappa_{1}\neq\kappa_{2}$ that satisfy~$\kappa_{1}\preceq\kappa_{2}\preceq\kappa_{1}$ (when $\kappa_{1}$ arises from $\kappa_{2}$ by permuting the output alphabet). However, for our purposes such channels can be considered as equivalent.  We write $\kappa_{1}\prec\kappa_{2}$ if $\kappa_{1}\preceq\kappa_{2}$ and $\kappa_{1}\not\succeq\kappa_{2}$.
By Blackwell's theorem this implies that $\kappa_{2}$ performs at least as good as $\kappa_{1}$ in any decision problem and that there exist decision problems in which $\kappa_{2}$ outperforms~$\kappa_{1}$.

For a given distribution of~$S$, we can also compare $\kappa_{1}$ and $\kappa_{2}$ by comparing the two mutual
informations $I(S;X_{1})$, $I(S;X_{2})$ between the common input $S$ and the channel outputs $X_{1}$ and $X_{2}$. The data processing inequality shows that $\kappa_{2}\succeq\kappa_{1}$ implies $I(S;X_{2})\ge I(S;X_{1})$.  However, the converse implication does not hold. 
The intuitive reason is that for the Blackwell order, 
not only the amount of information is important.  Rather, the question is how much of the information that $\kappa_{1}$
or $\kappa_{2}$ preserve is relevant for a given fixed decision problem (that is, a given fixed utility function).

Given two channels $\kappa_{1},\kappa_{2}$, suppose that $I(S;X_{2})\ge I(S;X_{1})$ for \emph{all} distributions of~$S$. In this case, we say that $\kappa_{2}$ is \emph{more capable} than~$\kappa_{1}$.  Does this imply that $\kappa_{1}\preceq\kappa_{2}$?  The answer is known to be negative in general~\cite{KoernerMarton75:Comparison_of_noisy_channels}.  
In Proposition~\ref{prop:channels} we introduce a new surprising example of this phenomenon with a particular structure.  In fact, in this example, $\kappa_{1}$ is a Markov approximation of~$\kappa_{2}$ by a deterministic function, in the following sense:
Consider another random variable $f(S)$ that arises from $S$ by applying a (deterministic) function~$f$.
% :\Scal\to\Scal'$ from the support $\Scal:=\big\{s:P_S(s)=\textsf{Pr}(S=s)>0\big\}$ of $S$ to another set $\Scal'$.
Given two random variables $S$, $X$, denote by $X\ot S$ the channel defined by the conditional probabilities $P_{X|S}(x|s)$, % on~$\Scal$.
and let $\kappa_{2}:=(X\ot S)$ and $\kappa_{1}:=(X\ot f(S))\cdot(f(S)\ot S)$.
Thus, $\kappa_1$ can be interpreted as first replacing~$S$ by~$f(S)$ and then sampling $X$ according to the conditional distribution~$P_{X|S}(x|f(s))$.
Which channel is superior? Using the data processing inequality, it is easy to see that $\kappa_{1}$ is less capable than~$\kappa_{2}$. However, as Proposition~\ref{prop:channels} shows, in general~$\kappa_{1}\not\preceq\kappa_{2}$.

We call $\kappa_{1}$ a Markov approximation, because the output of $\kappa_{1}$ is independent of the input~$S$ given~$f(S)$.  The channel $\kappa_{1}$ can also be obtained from $\kappa_{2}$ by ``pre-garbling'' (Lemma~\ref{lem:ex-as-pregarbling}); that is, 
there is another stochastic matrix~$\lambda^{f}$ that satisfies $\kappa_{1} = \kappa_{2}\cdot\lambda^{f}$.
It is known that pre-garbling may improve the performance of a channel (but not its capacity) as we recall in Section~\ref{sec:pregarbling}. What may be surprising is that this can happen for pre-garblings of the form~$\lambda^{f}$, which have the effect of coarse-graining according to~$f$.

The fact that the more capable preorder does not imply the Blackwell order shows that ``Shannon information,'' as captured by the mutual information, is not the same as ``Blackwell information,'' as needed for the Blackwell decision problems. Indeed, our example explicitly shows that even though coarse-graining always reduces Shannon information, it need not reduce Blackwell information.
Finally, let us mention that there are further ways of comparing channels (or stochastic matrices); see~\cite{CohenKempermanZbaganu98:Comparison_of_Stochastic_Matrices} for  an overview.

\smallskip

Proposition~\ref{prop:channels} builds upon another effect that we find paradoxical: Namely, there exist random variables $S,X_{1},X_{2}$ and there exists a function~$f:\Scal\to\Scal'$ from the support of~$S$ to a finite set~$S'$
such that the following holds:
\begin{enumerate}
\item $S$ and $X_{1}$ are independent given~$f(S)$.
\item $(X_{1}\ot f(S)) \;\preceq\; (X_{2}\ot f(S))$.
\item $(X_{1}\ot S) \;\not\preceq\; (X_{2}\ot S)$.
\end{enumerate}
Statement 1) says that everything $X_{1}$ knows about~$S$, it knows through $f(S)$.  Statement 2) says that $X_{2}$
knows more about $f(S)$ than~$X_{1}$.  Still, 3) says that we cannot conclude that $X_{2}$ knows more about~$S$
than~$X_{1}$.  The paradox illustrates that it is difficult to formalize what it means to ``know more.''

\bigskip

Understanding the Blackwell order is an important aspect of understanding information decompositions; that is, the
quest to find new information measures that separate different aspects of the mutual information
$I(S;X_{1},\dots,X_{k})$ of $k$ random variables $X_{1},\dots,X_{k}$ and a target variable~$S$ (see the other
contributions of this special issue and references therein).  In particular, \cite{BROJA13:Quantifying_unique_information}
argues that the Blackwell order provides a natural criterion when a variable~$X_{1}$ has unique information
about~$S$ with respect to~$X_{2}$.  We hope that the examples we present here are useful in developing intuition on how
information can be shared among random variables and how it behaves when applying a deterministic function, such as a
coarse-graining.  Further implications of our examples on information decompositions are discussed in~\cite{Everyone17:Extractable}.
In the converse direction, information decomposition measures (such as measures of unique information) can be used to
study the Blackwell order and deviations from the Blackwell order.  We illustrate this idea in Example~\ref{ex:main_hm}.

\bigskip

The remainder of this work is organized as follows: In Section~\ref{sec:pregarbling}, we recall how pre-garbling can be used to improve the performance of a channel.  We also show that the pre-garbled channel will always be less capable and that simultaneous pre-garbling of both channels preserves the Blackwell order. In Section~\ref{sec:pregarbling-by-condchan}, we state a few properties of the Blackwell order, and we explain why we find these properties counter-intuitive and paradoxical. In particular, we show that coarse-graining the input can improve the performance of a channel. Section~\ref{sec:example} contains a detailed discussion of an example that illustrates these properties.
In Section~\ref{sec:deficiency} we use the unqiue information measure from~\cite{BROJA13:Quantifying_unique_information}, which has properties similar to the Le Cam's deficiency, to illustrate deviations from the Blackwell relation.

\section{Pre-garbling}
\label{sec:pregarbling}

As discussed above (and as made formal in Blackwell's theorem (Theorem~\ref{thm:Blackwell})), garbling the output of a channel (``post-garbling'') never increases the quality of a channel.  
On the other hand, garbling the input of a channel (``pre-garbling'') may increase the performance of a channel, as the following example shows. 

\begin{example}
  \label{ex:pregarbling}
  Suppose that an agent can choose an action from a finite set~$\Acal$. She then receives a utility~$u(a,s)$ that depends both on the chosen action~$a\in\Acal$ and on the value $s$ of a random variable~$S$. Consider the channels
  \begin{equation*}
    \kappa_{1} =
    \begin{pmatrix}
      0.9 & 0 \\
      0.1 & 1
    \end{pmatrix}
    \text{ and }
    \kappa_{2} =
    \kappa_{1}\cdot
    \begin{pmatrix}
      0 & 1 \\
      1 & 0
    \end{pmatrix}
    = 
    \begin{pmatrix}
      0 & 0.9 \\
      1 & 0.1 
    \end{pmatrix},
\end{equation*}
and the utility function
\begin{center}
  \begin{tabular}{crrrr}
%    \toprule
    $s$      & 0 & 0 & 1 & 1 \\
    $a$      & 0 & 1 & 0 & 1 \\
    \midrule
    $u(s,a)$ & 2 & 0 & 0 & 1 \\
%    \bottomrule
  \end{tabular}
\end{center}
For uniform input the optimal decision rule for $\kappa_{1}$ is
\begin{equation*}
a(0) = 0, \;
a(1) = 1
\end{equation*}
and the opposite 
\begin{equation*}
a(0) = 1, \;
a(1) = 0
\end{equation*}
for $\kappa_{2}$. The expected utility with $\kappa_{1}$ is~$1.4$, while using~$\kappa_{2}$, it is slightly higher, $1.45$.

It is also not difficult to check that neither of the two channels is a garbling of the other (cf.~Prop.~3.22
in~\cite{CohenKempermanZbaganu98:Comparison_of_Stochastic_Matrices}).
\end{example}

The intuitive reason for the difference in the expected utilities is that the channel $\kappa_{2}$ transmits one of the states without noise and the other state with noise.  With a convenient pre-processing, it is possible to make sure that the relevant information \emph{for choosing an action and for optimizing expected utility} is transmitted with less noise.

Note the symmetry of the example: Each of the two channels arises from the other by a convenient pre-processing, since the pre-processing is invertible.  Hence, the two channels are not comparable by the Blackwell order. In contrast, two channels that only differ by an invertible garbling of the output are equivalent with respect to the Blackwell order.

The pre-garbling in Example~\ref{ex:pregarbling} is invertible, and so it is more aptly described as a pre-processing.
In general, though, pure pre-garbling and pure pre-processing are not easily distinguishable,
and it is easy to perturb 
Example~\ref{ex:pregarbling} by adding noise without changing the conclusion. In Section~\ref{sec:pregarbling-by-condchan}, we will present an example in which the pre-garbling consists of
coarse-graining. It is much more difficult to understand how coarse-graining can be used as sensible pre-processing.

\medskip

Even though pre-garbling can make a channel better (or, more precisely, more suited for a particular decision problem at hand), pre-garbling cannot invert the Blackwell order:
\begin{lemma}
  If $\kappa_{1} \prec \kappa_{2}\cdot\lambda$, then $\kappa_{1}\not\succeq\kappa_{2}$.
\end{lemma}
\begin{proof}
  Suppose that $\kappa_{1} \prec \kappa_{2}\cdot\lambda$.  Then the capacity of $\kappa_{1}$ is less than the capacity
  of~$\kappa_{2}\cdot\lambda$, which is bounded by the capacity of~$\kappa_{2}$.  Therefore, the capacity
  of~$\kappa_{1}$ is less than the capacity of~$\kappa_{2}$.
\end{proof}

Also, it follows directly from Blackwell's theorem that
\begin{equation*}
  \kappa_{1}\preceq\kappa_{2}
  \text{ implies }
  \kappa_{1}\cdot\lambda\preceq\kappa_{2}\cdot\lambda
\end{equation*}
for any channel~$\lambda$, where the input and output alphabets of $\lambda$ equal the input alphabet
of~$\kappa_{1},\kappa_{2}$.  Thus, pre-garbling preserves the Blackwell order when applied to both channels
simultaneously.

Finally, let us remark that certain kinds of simultaneous pre-garbling can also be ``hidden'' in the utility function: Namely, in Blackwell's theorem, it is not necessary to vary the distribution of~$S$, as long as the support of the (fixed) input distribution has full support~$\Scal$ (that is, every state of the input alphabet of $\kappa_{1}$ and $\kappa_{2}$ appears with positive probability).  In this setting, it suffices to look only at different utility functions.  When the input distribution is fixed, it is more convenient to think in terms of random variables instead of channels,
which slightly changes the interpretation of the decision problem.  Suppose we are given random variables
$S,X_{1},X_{2}$ and a utility function $u(a,s)$ depending on the value of $S$ and an action~$a\in\Acal$ as above.  If we cannot look at both $X_{1}$ and~$X_{2}$, should we rather look at $X_{1}$ or at $X_{2}$ to take our decision?

\begin{theorem}[Blackwell's theorem for random variables~\cite{BROJA13:Quantifying_unique_information}]
  The following two conditions are equivalent:
  \begin{enumerate}
  \item Under the optimal decision rule, when the agent chooses $X_{2}$, her expected utility is always at least as big as the expected utility when she chooses~$X_{1}$, independent of the utility function.
  \item $(X_{1}\ot S) \preceq (X_{2}\ot S)$.
  \end{enumerate}
\end{theorem}

\section{Pre-garbling by coarse-graining}
\label{sec:pregarbling-by-condchan}

In this section we present a few counter-intuitive properties of the Blackwell order.

\begin{proposition}
  \label{prop:rv}
  There exist random variables $S,X_{1},X_{2}$ and a function~$f:\Scal\to\Scal'$ from the support of~$S$ to
  a finite set~$S'$ such that the following holds:
  \begin{enumerate}
  \item $S$ and $X_{1}$ are independent given~$f(S)$.
  \item $(X_{1}\ot f(S)) \;\prec\; (X_{2}\ot f(S))$.
  \item $(X_{1}\ot S) \;\not\preceq\; (X_{2}\ot S)$.
  \end{enumerate}
\end{proposition}
This result may at first seem paradoxical. After all, property~3) implies that there exists a decision problem involving $S$ for which it is better to use $X_1$ than $X_2$. 
Property~1) implies that any
information that~$X_{1}$ has about~$S$ is contained in $X_{1}$'s information about~$f(S)$.
One would therefore expect that, from the viewpoint of~$X_{1}$, any decision problem in which the task is to predict~$S$ and to react on~$S$ looks like a decision problem in which the task is to react to~$f(S)$.  
But property~2) implies that for such a decision problem, it may in fact be better to look at~$X_{2}$.

\begin{proof}[Proof of Proposition~\ref{prop:rv}]
  The proof is by Example~\ref{ex:main}, which will be given in Section~\ref{sec:example}.  This example satisfies
  \begin{enumerate}
  \item $S$ and $X_{1}$ are independent given~$f(S)$.
  \item $(X_{1}\ot f(S)) \;\preceq\; (X_{2}\ot f(S))$.
  \item $(X_{1}\ot S) \;\not\preceq\; (X_{2}\ot S)$.
  \end{enumerate}
  It only remains to show that it is possible to also achieve the strict relation $(X_{1}\ot f(S)) \prec (X_{2}\ot
  f(S))$ in the second statement.  This can easily be done by adding a small garbling to the channel $X_{1}\ot f(S)$
  (e.g.~by adding a binary symmetric channel with sufficiently small noise parameter~$\epsilon$).  This ensures
  $(X_{1}\ot f(S))\prec(X_{2}\ot f(S))$, and if the garbling is small enough, this does not destroy the property
  $(X_{1}\ot S) \;\not\preceq\; (X_{2}\ot S)$.
\end{proof}

The example from Proposition~\ref{prop:rv} also leads to the following paradoxical property:
\begin{proposition}
  \label{prop:channels}
  There exist random variables $S,X$ and there exists a function~$f:\Scal\to\Scal'$ from the support of~$S$ to
  a finite set~$S'$ such that the following holds:
  \begin{equation*}
    (X\ot f(S))\cdot(f(S)\ot S) \;\not\preceq\; X\ot S.
  \end{equation*}
\end{proposition}
Let us again give a heuristic argument why we find this property paradoxical.  Namely, the combined channel $(X\ot
f(S))\cdot(f(S)\ot S)$ can be seen as a Markov chain approximation of the direct channel~$X\ot S$ that corresponds to replacing the conditional distribution
\begin{equation*}
  P_{X|S}(x|s) =  \sum_{f(s)}  P_{X|Sf(S)}(x|s,f(s)) P_{f(S)|S}(f(s)|s).
\end{equation*}
by
\begin{equation*}
  \sum_{f(s)}  P_{X|f(S)}(x|f(s))P_{f(S)|S}(f(s)|s).
\end{equation*}
Proposition~\ref{prop:channels} together with Blackwell's theorem states that there exist situations where this
approximation is better than the correct channel.

\begin{proof}[Proof of Proposition~\ref{prop:channels}]
  Let $S,X_{1},X_{2}$ be as in Example~\ref{ex:main} in Section~\ref{sec:example} that also proves
  Proposition~\ref{prop:rv}, and let~$X=X_{2}$. In that example, the two channels $X_{1}\ot f(S)$ and $X_{2}\ot f(S)$ are equal. Moreover, $X_{1}$ and $S$ are independent given~$f(S)$.  Thus, $(X\ot f(S))\cdot(f(S)\ot S) = (X_{1}\ot S)$.  Therefore, the statement follows from $(X_{1}\ot S)\not\preceq(X_{2}\ot S)$.
\end{proof}

On the other hand, the channel $(X\ot f(S))\cdot(f(S)\ot S)$ is always less capable than $X\ot S$:
\begin{lemma}
  \label{lemma:capable}
  For any random variables $S$, $X$, and function~$f:\Scal\to\Scal$, the channel $(X\ot f(S))\cdot(f(S)\ot S)$ is less capable than~$X\ot S$.
\end{lemma}
\begin{proof}
  For any distribution of~$S$, let $X'$ be the output of the channel $(X\ot f(S))\cdot(f(S)\ot S)$.  Then, $X'$ is
  independent of $S$ given~$f(S)$.  On the other hand, since $f$ is a deterministic function, $X'$ is independent of $f(S)$ given~$S$.  Together, this implies $I(S;X') = I(f(S);X')$.  Using the fact that the joint distributions of $(X,f(S))$ and $(X',f(S))$ are identical and applying the data processing inequality gives
  \begin{equation*}
    I(S;X') = I(f(S);X') = I(f(S);X) \le I(S;X). \qedhere
  \end{equation*}
\end{proof}

The setting of Proposition~\ref{prop:channels} can also be understood as a specific kind of pre-garbling.  Namely,
consider the channel $\lambda^{f}$ defined by
\begin{equation*}
  \lambda^{f}_{s',s} := P_{S|f(S)}(s'|f(s)).
\end{equation*}
The effect of this channel can be characterized as a randomization of the input: The precise value of~$S$ is forgotten, and only the value of $f(S)$ is
preserved.  Then a new value $s'$ is sampled for~$S$ according to the conditional distribution of~$S$ given $f(S)$.
\begin{lemma}
  \label{lem:ex-as-pregarbling}
  $(X\ot f(S))\cdot(f(S)\ot S) = (X\ot S)\cdot\lambda^{f}$.
\end{lemma}
\begin{proof}
  $\displaystyle\sum_{s_{1}}P_{X|S}(x|s_{1})P_{S|f(S)}(s_{1}|f(s)) = \sum_{s_{1}, t}P_{X|S}(x|s_{1})P_{S|f(S)}(s_{1}|t)P_{f(S)|S}(t|s)$
  \begin{equation*}
%    &\sum_{s_{1}}P(X=x|S=s_{1}) P(S=s_{1}|f(S) = f(s)) \\
%    &\\
    = \sum_{t}P_{X|f(S)}(x|t)P_{f(S)|S}(t|s),
  \end{equation*}
  where we have used that $X-S-f(S)$ forms a Markov chain.
\end{proof}
While it is easy to understand that pre-garbling can be advantageous in general (since it can work as
preprocessing), we find surprising that this can also happen in the case where the pre-garbling is done in terms of a
function~$f$; that is, in terms of a channel $\lambda^{f}$ that does coarse-graining.

\section{Examples}
\label{sec:example}

\begin{example}
  \label{ex:main}
  Consider the joint distribution
  \begin{center}
    \begin{tabular}{cccc|c}
      $f(s)$ & $s$&$x_{1}$&$x_{2}$& $P_{f(S)SX_1X_2}$ \\
      \hline
      0 & 0 & 0 & 0 & 1/4 \\
      % 0 & 1 & 0 & 0 & 0 \\
      % 0 & 0 & 0 & 1 & 0 \\
      0 & 1 & 0 & 1 & 1/4 \\
      0 & 0 & 1 & 0 & 1/8 \\
      0 & 1 & 1 & 0 & 1/8 \\
      %\hline
      1 & 2 & 1 & 1 & 1/4 \\
    \end{tabular}
  \end{center}
  and the function $f:\{0,1,2\}\to\{0,1\}$ with $f(0)=f(1)=0$ and $f(2)=1$.  Then $X_{1}$ and $X_{2}$ are independent uniform   binary random variables, and $f(S) = \textsc{And}(X_{1},X_{2})$. By symmetry, the joint distributions of the pairs $(f(S), X_{1})$ and $(f(S), X_{2})$ are identical, and so the two channels $X_{1}\ot f(S)$ and $X_{2}\ot f(S)$ are identical. In particular $(X_{1}\ot f(S))\preceq(X_{2}\ot f(S))$.

  On the other hand, consider the utility function
  \begin{center}
    \begin{tabular}{cc|c}
      $s$&$a$&$u(s,a)$ \\
      \hline
      0 & 0 & 0 \\
      0 & 1 & 0 \\
      %\hline
      1 & 0 & 1 \\
      1 & 1 & 0 \\
      %\hline
      2 & 0 & 0 \\
      2 & 1 & 1 \\
    \end{tabular}
  \end{center}
  To compute the optimal decision rule, let us look at the conditional distributions:
  \begin{center}
    \begin{tabular}{cc|c}
      $s$&$x_{1}$& $P_{S|X_1}(s|x_{1})$\\
      \hline
      0 & 0 & 1/2 \\
      1 & 0 & 1/2 \\
      %\hline
      0 & 1 & 1/4 \\
      1 & 1 & 1/4 \\
      2 & 1 & 1/2 \\
    \end{tabular}
    $\quad$
    \begin{tabular}{cc|c}
      $s$&$x_{2}$& $P_{S|X_2}(s|x_{2})$\\
      \hline
      0 & 0 & 3/4 \\
      1 & 0 & 1/4 \\
      %\hline
      0 & 1 &  0  \\
      1 & 1 & 1/2 \\
      2 & 1 & 1/2 \\
    \end{tabular}
  \end{center}
  The optimal decision rule for $X_{1}$ is $a(0) = 0$, $a(1) = 1$, with expected utility
  \begin{equation*}
    u_{X_{1}} := 1/2\cdot 1/2 + 1/2\cdot 1/2 = 1/2.
  \end{equation*}
  The optimal decision rule for $X_{2}$ is $a(0) = 0$, $a(1) \in \{0,1\}$ (this is not unique in this case), with expected utility
  \begin{equation*}
    u_{X_{2}} := 1/2\cdot 1/4 + 1/2\cdot 1/2 = 3/8 < 1/2.
  \end{equation*}
\end{example}

  How can we understand this example?
  Some observations:
  \begin{itemize}
  \item It is easy to see that~$X_{2}$ has more irrelevant information than~$X_{1}$: namely, $X_{2}$ can determine relatively precisely when~$S=0$. However, since $S=0$ gives no utility independent of the action, this information is not relevant.
  It is more difficult to understand why~$X_{2}$ has less relevant information than~$X_{1}$. Surprisingly, $X_{1}$ can determine more precisely when~$S=1$: If~$S=1$, then~$X_{1}$ ``detects this'' (in the sense that~$X_{1}$ chooses action~$0$) with probability~$2/3$. For $X_{2}$, the same probability is only~$1/3$.
  \item The conditional entropies of~$S$ given~$X_{2}$ are smaller than the conditional entropies of~$S$ given~$X_{1}$:
    \begin{align*}
      H(S|X_{1}=0) &= \log(2), &H(S|X_{1}=1) &= \textstyle\frac32\log(2), \\
      H(S|X_{2}=0) &= 2\log(2) - \textstyle\frac32\log(3) \approx
      0.4150375\log(2), &H(S|X_{2}=1) &= \log(2).
    \end{align*}
  \item One can see in which sense~$f(S)$ captures the relevant information for~$X_{1}$, and indeed for the whole
    decision problem: knowing~$f(S)$ is completely sufficient in order to receive the maximal utility for each state of~$S$. However, when information is incomplete, it matters how the information about the different states of~$S$ is mixed, and two variables $X_{1},X_{2}$ that have the same joint distribution with~$f(S)$ may perform differently. 
    It is somewhat surprising that it is the random variable that has less information about~$S$ and that is conditionally independent of~$S$ given~$f(S)$ which actually performs better.
\end{itemize}

Example~\ref{ex:main} is different from the pre-garbling Example~\ref{ex:pregarbling} discussed in
Section~\ref{sec:pregarbling}. In the latter, both channels had the same amount of information (mutual
information) about~$S$, but for the given decision problem the information provided by~$\kappa_{2}$ was more relevant than the information provide by~$\kappa_1$.
The first difference in Example~\ref{ex:main} is that $X_1$ has less mutual information about~$S$ than~$X_2$ (Lemma~\ref{lemma:capable}). Moreover, both channels are identical with respect to~$f(S)$, i.e. they provide
the same information about~$f(S)$, and for~$X_1$ it is the only information it has about~$S$. So, one could argue that~$X_2$ has \emph{additional} information, that does not help though, but \emph{decreases} the expected utility instead. 

\bigskip

We give another example which shows that $X_{2}$ can also be chosen a deterministic function of~$S$.

\begin{example}
  \label{ex:main-detX2}
  Consider the joint distribution
  \begin{center}
    \begin{tabular}{cccc|c}
      $f(s)$ & $s$ & $x_{1}$&$x_{2}$& $P_{f(S)SX_1X_2}$ \\
      \hline
      0 & 0 & 0 & 0 & 1/6 \\
      0 & 0 & 1 & 0 & 1/6 \\
      0 & 1 & 0 & 1 & 1/6 \\
      0 & 1 & 1 & 1 & 1/6 \\
      %\hline
      1 & 2 & 1 & 1 & 1/3 \\
    \end{tabular}
  \end{center}
  The function $f$ is as above, but now also $X_{2}$ is a function of~$S$.  Again, the two channels $X_{1}\ot f(S)$ and
  $X_{2}\ot f(S)$ are identical, and $X_{1}$ is independent of $S$ given~$f(S)$.
  Consider the utility function
  \begin{center}
    \begin{tabular}{cc|c}
      $s$&$a$&$u(s,a)$ \\
      \hline
      0 & 0 & 0 \\
      0 & 1 & 0 \\
      %\hline
      1 & 0 & 0 \\
      1 & 1 & 1 \\
      %\hline
      2 & 0 & 0 \\
      2 & 1 & $-1$ \\
    \end{tabular}
  \end{center}
  One can show that it is optimal for agent who relies on $X_{2}$ to always choose action~0, which brings no reward (and
  no loss).  However, when the agent knows that $X_{1}$ is zero, he may safely choose action~1 and has a positive
  probability of receiving a positive reward.
\end{example}

To add another interpretation to the last example, we visualize the situation in the following Bayesian network:
\begin{equation*}
  X \leftarrow S \rightarrow f(S) \rightarrow X',
\end{equation*}
where, as in Proposition~\ref{prop:channels} and its proof, we let $X = X_{2}$, and we consider $X' = X_{1}$ as an
approximation of~$X$.  Then $S$ denotes the state of the system that we are interested in, and $X$ denotes a given set
of observables of interest.  $f(S)$ can be considered as a ``proxy'' in situations where it is difficult to observe~$X$
directly.  For example, in neuroimaging, instead of directly measuring the neural activity~$X$, one might look at an MRI
signal~$f(S)$.  In economic and social sciences, monetary measures like the GDP are used as a proxy for prosperity.
% , for instance, a set of measurements on our system or a set of given features in a classification
% problem.

A decision problem can always be considered as a classification problem defined by the utility $u(s,a)$ by
considering the optimal action as the class label of state $S$.  Proposition~\ref{prop:channels} now says that there
exist $S,X,f(S)$ and a classifcation problem $u(s,a)$, such that the approximated features $X'$ (simulated from~$f(S)$)
allow for a better classification (higher utility) than the original features~$X$.

In such a situation, looking at $f(S)$ will always be better than looking at either $X$ or~$X'$.  Thus, the paradox will
only play a role in situations where it is not possible to base the decision on $f(S)$ directly.  For example, $f(S)$
might still be too large, or $X$ might have a more natural interpretation, making it easier to interpret for the
decision taker. But, when it is better to base a decision on a proxy rather than directly on the observable of interest, this interpretation may be erroneous.

\section{Information decomposition and Le Cam deficiency}
\label{sec:deficiency}

Given two channels~$\kappa_{1},\kappa_{2}$, how can one decide whether or not $\kappa_{1}\preceq\kappa_{2}$?  The easiest way is to check whether the equation $\kappa_{1} = \lambda\cdot\kappa_{2}$ has a solution~$\lambda$ that is a stochastic matrix.  In the finite alphabet case, this amounts to checking feasibility of a linear program, which is considered computationally easy.
However, when the feasibility check returns a negative result, this approach does not give any more information, e.g. how far $\kappa_{1}$ is away from being a garbling of~$\kappa_{2}$.
A function that quantifies how far $\kappa_{1}$ is away from being a garbling of~$\kappa_{2}$ is given by the \emph{(Le Cam) deficiency} and its various generalizations \cite{Raginsky2011}.
Another such function is given by $UI$ defined in~\cite{BROJA13:Quantifying_unique_information} that takes into account that the channels we consider are of the form $\kappa_{1}=(X_{1}\ot S)$ and $\kappa_{2}=(X_{2}\ot S)$, that is, they are derived from conditional distributions of random variables. In contrast to the deficiencies, $UI$ depends on the input distribution to these channels.

Let $P_{SX_1X_2}$ be a joint distribution of $S$ and the outputs $X_{1}$ and~$X_{2}$.
Let $\Delta_{P}$ be the set of all joint distributions of the random variables $S,X_{1},X_{2}$ (with the same alphabets) that are compatible with the marginal distributions of $P_{SX_1X_2}$ for the pairs $(S,X_{1})$ and $(S,X_{2})$, i.e.,
\begin{equation*}
  \Delta_{P} := \big\{Q_{SX_1X_2} \in\Delta: Q_{SX_1}=P_{SX_1},  Q_{SX_2}=P_{SX_2} \big\}.
\end{equation*}
In other words, $\Delta_{P}$ consists of all joint distributions that are compatible with~$\kappa_{1}$ and~$\kappa_{2}$ and that have the same distribution for~$S$ as~$P_{SX_1X_2}$. Consider the function
\begin{align*}
UI(S;X_1\backslash X_2) := \min_{Q\in\Delta_{P}} I_{Q}(S;X_1|X_2),
\end{align*}
where $I_{Q}$ denotes the conditional mutual information evaluated with respect to the the joint distribution~$Q$. This function has the following property: $UI(S;X_{1}\backslash X_{2}) = 0$ if and only if $\kappa_{1}\preceq\kappa_{2}$~\cite{BROJA13:Quantifying_unique_information}.
Computing $UI$ is a convex optimization problem. However, the condition number can be very bad, which makes the problem difficult in practice.

$UI$ is interpreted in~\cite{BROJA13:Quantifying_unique_information} as a measure of the \emph{unique} information that $X_1$ conveys about $S$ (with respect to $X_2$).
So, for instance, with this interpretation Example \ref{ex:main} can be summarized as follows: Neither $X_1$ nor $X_2$ has unique information about~$f(S)$.  However, both variables have unique information about~$S$, although $X_1$ is conditionally independent of $S$ given $f(S)$ and thus, in contrast to $X_2$, contains no ``additional'' information about~$S$.
We now apply $UI$ to a parameterized version of the \textsc{And} gate in Example \ref{ex:main}.

\begin{example}
  \label{ex:main_hm}
  Figure \ref{fig:fig_UI_a} shows a heat map of $UI$ computed on the set of all distributions
  of the form
\begin{center}
  \begin{tabular}{cccc|c}
    $f(s)$ & $s$&$x_{1}$&$x_{2}$& $P_{f(S)SX_1X_2}$ \\
    \hline
    0 & 0 & 0 & 0 & $\nicefrac{1}{8}+2b$ \\
    0 & 1 & 0 & 0 & $\nicefrac{1}{8}-2b$ \\
    0 & 0 & 0 & 1 & $\nicefrac{1}{8}+a$ \\
    0 & 1 & 0 & 1 & $\nicefrac{1}{8}-a$ \\
    0 & 0 & 1 & 0 & $\nicefrac{1}{8}+\nicefrac{a}{2}+b$ \\
    0 & 1 & 1 & 0 & $\nicefrac{1}{8}-\nicefrac{a}{2}-b$ \\
    %\hline
    1 & 2 & 1 & 1 & $\nicefrac{1}{4}$ % \\
  \end{tabular}
\end{center}
where $-\nicefrac{1}{8} \leq a \leq \nicefrac{1}{8}$ and $-\nicefrac{1}{16} \leq b \leq \nicefrac{1}{16}$.
This is the set of distributions of $S,X_{1},X_{2}$ that satisfy the following constraints:
\begin{enumerate}
\item $X_{1},X_{2}$ are independent;
\item $f(S) = \textsc{And}(X_1,X_2)$, where $f$ is as in Example~\ref{ex:main}; and
\item $X_{1}$ is independent of~$S$ given~$f(S)$.
\end{enumerate}
Along the secondary diagonal $b=\nicefrac{a}{2}$, the marginal distributions of the pairs $(S,X_1)$ and $(S,X_2)$ are identical.  In such a situation, the channels $(X_{1}\ot S)$ and $(X_{2}\ot S)$ are Blackwell-equivalent, and so $UI$ vanishes.  Furtheraway from the diagonal, the marginal distributions differ, and $UI$ grows.  The maximum value is achieved at the corners for $(a,b)=(-\nicefrac{1}{8},\nicefrac{1}{16})$, $(\nicefrac{1}{8},-\nicefrac{1}{16})$.  At the upper left corner $(a,b)=\pm(-\nicefrac{1}{8},\nicefrac{1}{16})$, we recover Example~\ref{ex:main}.

\begin{figure}
  	\centering
  	\begin{subfigure}{.5\textwidth}
  	\centering
  	\begin{tikzpicture}
  	\begin{axis}[
  	axis on top,
  	title={$UI$},
  	xlabel={$a$}, ylabel={$b$}, 
  	ylabel style={rotate=-90}, 
  	width=0.75\linewidth, height=0.75\linewidth,
  	xmin=-0.1213, xmax=0.1213, 
  	ymin=-0.0620, ymax=0.0620, 
  	xticklabel style={
  		/pgf/number format/fixed,
  		/pgf/number format/precision=5},
  	yticklabel style={
  		/pgf/number format/fixed,
  		/pgf/number format/precision=5},
  	colormap/autumn,
  	colorbar,
  	colorbar style={at={(1.1,0.001)},anchor=south}, 
   	colorbar style={
  		point meta min=0, 
  		point meta max=0.17, 
  		yticklabel={
  			\num[
  			scientific-notation = fixed,
  			round-precision = 3,
  			]{\tick}
  		}}
  	] 
  	\addplot graphics [xmin=-0.1213,xmax=0.1213,ymin=-0.0620,ymax=0.0620] {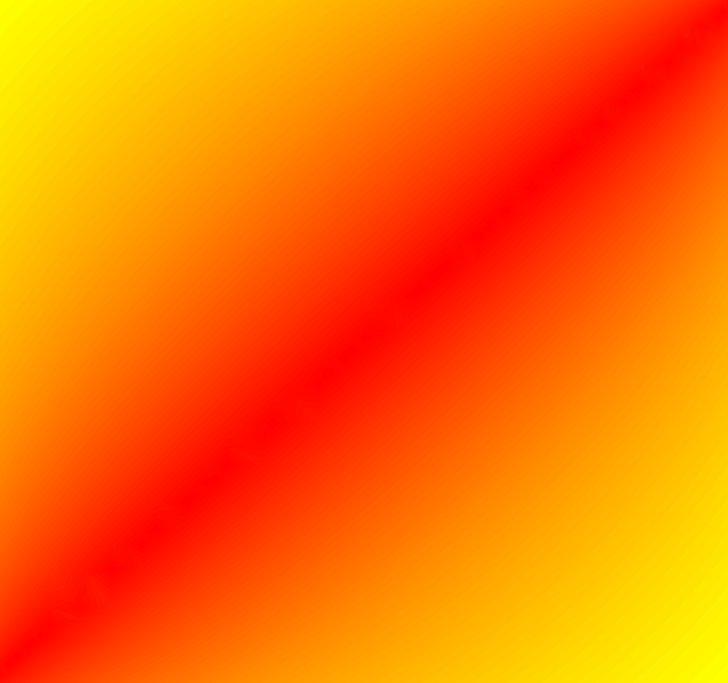};
  	\end{axis}
   	\end{tikzpicture}
  	\caption{}
  	\label{fig:fig_UI_a}
    \end{subfigure}%
   \begin{subfigure}{.5\textwidth}
   	\centering
   	\begin{tikzpicture}
   	\begin{axis}[
   	axis on top,
   	title={$UI$},
   	xlabel={$a$}, ylabel={$b$}, 
   	ylabel style={rotate=-90}, 
   	width=0.75\linewidth, height=0.75\linewidth,
   	xmin=0, xmax=1, 
   	ymin=0, ymax=1, 
   	xticklabel style={
   		/pgf/number format/fixed,
   		/pgf/number format/precision=5},
   	yticklabel style={
   		/pgf/number format/fixed,
   		/pgf/number format/precision=5},
   	colormap/autumn,
   	colorbar,
   	colorbar style={at={(1.1,0.001)},anchor=south}, 
   	colorbar style={
   		point meta min=0, 
   		point meta max=0.19, 
   		yticklabel={
   			\num[
   			scientific-notation = fixed,
   			round-precision = 3,
   			]{\tick}
   		}}
   	] 
   	\addplot graphics [xmin=0,xmax=1,ymin=0,ymax=1] {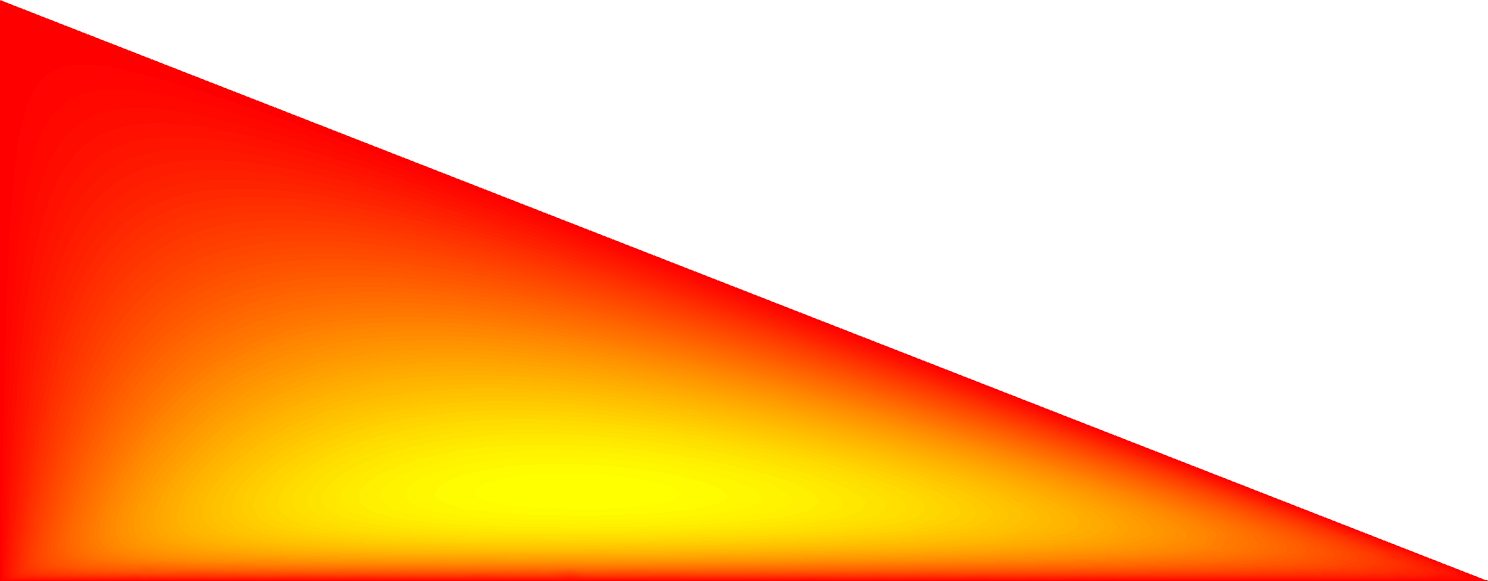};
   	\end{axis}
   	\end{tikzpicture}
   	\caption{}
   	\label{fig:fig_UI_b}
   \end{subfigure}
   \caption{Heatmaps for the function $UI$ in (a) Example \ref{ex:main_hm}, and (b) Example \ref{ex:main_hm_detX2}.}
   \label{fig:fig_UI}
\end{figure}
\end{example}

\begin{example}
  \label{ex:main_hm_detX2}
  Figure \ref{fig:fig_UI_b} shows a heat map of $UI$ computed on the set of all distributions of the form
  \begin{center}
    \begin{tabular}{cccc|c}
      $f(s)$ & $s$ & $x_{1}$&$x_{2}$& $P_{f(S)SX_1X_2}$ \\
      \hline
      0 & 0 & 0 & 0 & $\nicefrac{a^2}{(a+b)}$ \\
      0 & 0 & 1 & 0 & $\nicefrac{ab}{(a+b)}$ \\
      0 & 1 & 0 & 1 & $\nicefrac{ab}{(a+b)}$ \\
      0 & 1 & 1 & 1 & $\nicefrac{b^{2}}{(a+b)}$ \\
      %\hline
      1 & 2 & 1 & 1 & $1 - a - b$ \\
    \end{tabular}
  \end{center}
  where $a,b\ge 0$ and $a+b\le 1$.  This extends Example~\ref{ex:main-detX2}, which is recovered for $a=b=\nicefrac13$.
  This is the set of distributions of $S,X_{1},X_{2}$ that satisfy the following constraints:
  \begin{enumerate}
  \item $X_{2}$ is a function of~$S$, where the function is as in Example~\ref{ex:main-detX2}.
  \item $X_{1}$ is independent of~$S$ given~$f(S)$.
  \item The channels $X_{1}\ot f(S)$ and $X_{2}\ot f(S)$ are identical.
  \end{enumerate}
  %\jr{Add figure.}
\end{example}

% \bigskip

%%%%%%%%%%%%%%%%%%%%%%%%%%%%%%%%%%%%%%%%%%
%\acknowledgments{
%We thank the participants of the PID workshop at FIAS in Frankfurt in December 2016 for many stimulating discussions on this subject.
%Nils Bertschinger thanks Dr. h.c. Maucher for funding his position.}

%%%%%%%%%%%%%%%%%%%%%%%%%%%%%%%%%%%%%%%%%%
%\authorcontributions{ %
%  The research was initiated by JR and carried out by all authors.  Computer experiments to find and analyze the
%  examples were done by PKB.  DW simplified Example~\ref{ex:pregarbling}.  JJ and NB added interpretation. The
%  manuscript was written by JR, PKB, EO and DW.  All authors have read and approved the final manuscript. }
%%The following statements should be used ``X.X. and Y.Y. conceived and designed the experiments; X.X. performed the experiments; X.X. and Y.Y. analyzed the data; W.W. contributed reagents/materials/analysis tools; Y.Y. wrote the paper.'' Authorship must be limited to those who have contributed substantially to the work reported.}


%=====================================
% References, variant B: external bibliography
%=====================================
%\externalbibliography{yes}
%\bibliography{../Info}

\bibliographystyle{IEEEtran}
\bibliography{BadBlackwell}

%\bigskip

%%%%%%%%%%%%%%%%%%%%%%%%%%%%%%%%%%%%%%%%%%
%% optional
%\appendixtitles{yes} %Leave argument "no" if all appendix headings stay EMPTY (then no dot is printed after "Appendix A"). If the appendix sections contain a heading then change the argument to "yes".
%\appendixsections{multiple} %Leave argument "multiple" if there are multiple sections. Then a counter is printed ("Appendix A"). If there is only one appendix section then change the argument to "one" and no counter is printed ("Appendix").
%\appendix

\end{document}